\documentclass[a4paper]{article}

\usepackage[pages=all, color=black, position={current page.south}, placement=bottom, scale=1, opacity=1, vshift=5mm]{background}

\usepackage[margin=1in]{geometry} 

\usepackage{amsmath}
\usepackage{amsthm}
\usepackage{amssymb}
\usepackage{ bbold }
\usepackage{xcolor}
\usepackage{multirow}
\usepackage{booktabs}
\usepackage{tcolorbox}

\usepackage{listings}

\input{fitch.sty}

\usepackage[utf8]{inputenc}
\usepackage{hyperref}
\hypersetup{
	unicode,
	pdfauthor={Reid Dale},
	pdftitle={Use and Evaluation of Synthetic Data},
	pdfsubject={A simple article template},
	pdfkeywords={article, template, simple},
	pdfproducer={LaTeX},
	pdfcreator={pdflatex}
}

\usepackage[sort&compress,numbers,square]{natbib}
\theoremstyle{plain}
\newtheorem{theorem}{Theorem}

\theoremstyle{definition}
\newtheorem{definition}[theorem]{Definition}
\newtheorem{example}[theorem]{Example}

\usepackage{tikz}
\usetikzlibrary{matrix}

\usepackage{graphicx, color}
\graphicspath{{fig/}}

\usepackage{algorithm, algpseudocode} 
\usepackage{mathrsfs} 

\usepackage{lipsum}

\title{Synthetic Data, Information, and Prior Knowledge: Why Synthetic Data Augmentation to Boost Sample Doesn't Work for Statistical Inference}
\author{Reid Dale$^1$, Jordan Rodu$^2$,Mike Baiocchi$^3$}

\date{
	$^1$Stanford University School of Medicine Department of Cardiothoracic Surgery \\ \texttt{reiddale@stanford.edu}\\%
    $^3$University of Virginia Department of Statistics\\
    \texttt{jsr6q@virginia.edu}
    \\
    $^2$Stanford University Department of Epidemiology $\&$ Population Health \\ \texttt{baiocchi@stanford.edu}
}

\begin{document}
	\maketitle
	
	\begin{abstract}
    The use of synthetic data to deidentify data and to improve predictive models is well-attested to. The augmentation of datasets using synthetically generated data is an alluring proposition: in the best case, it generates realistic data \textit{in silico} at a fraction of the cost of authentic data which may be found \textit{in vivo} or \textit{in vitro}. This poses novel epistemic challenges.

    We contend that synthetic data augmentation is best understood as a novel way of accounting for prior knowledge. In this manuscript, we propose a definition of synthetic distributions and analyze how synthetic data augmentation interplays with standard accounts of  maximum likelihood and Bayesian estimation. We observe that the marginal Fisher information contributed by synthetic data processes is subject to fundamental bounds, and enumerate obstacles to the use of synthetic data augmentation to aid in inferential tasks. 

    We then articulate a Bayesian formulation of the way that synthetic data augmentation can be coherently understood, but argue that naive approaches to the specification of the prior are epistemically unjustifiable. This suggests that enhanced scrutiny must be placed on identifying justifiable priors to warrant the use and inclusion of data drawn from specific synthetic distributions.

    While our analysis shows the challenges and limitations of using synthetic data augmentation to improve upon traditional statistical model reasoning, it does suggest that augmentation is the principal approach analysts using outcome reasoning (i.e. using train/test splits to justify the analysis) can constrain an otherwise high-dimensional model space, providing an alternative to trying to encode the constraints into the potentially complex architecture of the algorithm.




	\end{abstract}

	\tableofcontents

\section{Introduction}

The generation of realistic-looking data enjoying similar statistical properties to some target real-world distribution is known as \textit{synthetic data generation}. There are a handful of distinct types of use cases for synthetic data in the literature. 

Perhaps its most commonly accepted use is in the generation of datasets $S$ from a dataset $X$ with the aim of limiting disclosure of sensitive information \cite{kim1986method}\cite{kim1995masking}, which we term a \textit{masking} of $X$. The construction of masked counterparts can be motivated by its ability to enable research teams to access realistic datasets on which statistical tasks may be performed while being compliant with privacy standards established by HIPAA, GDRP, among others. Perhaps most notably, the United States Census Bureau has been working on masking techniques for use with its American Community Survey (ACS) \cite{rodriguez2023confidentiality}\cite{freiman2018formal}\cite{reynolds2023synthetic}\cite{daily2022disclosure}. 

A different use case centers on getting human-level insight into the generation process. For example, Chatbot Arena is a platform wherein a synthetic datum can be generated, and a human can indicate their preference in responses \cite{chatbot_arena}. In this way, synthetic data can partially augment a dataset to improve model performance through the injection of human feedback as the source of novel information.

Another use of synthetic data can be to augment a dataset by boosting sample that conforms to some prior knowledge, such as the use of physics-informed synthetic data for applications in MRI reconstruction \cite{WANG2025103616}. In such applications, synthetic data is used to reduce the variance in the target model through constraining the space of models by weighting the architecture of a black-box model towards models that approximately accommodate expert knowledge of the underlying data generating process. In this way, synthetic data generation becomes a flexible way to infuse prior knowledge in a way that does not require the construction of custom model architectures. The ``assumptions'' enter through the synthetic data used to fit the model. As these examples illustrate, synthetic data use is both highly flexible and stratified in its use cases.

We draw a distinction between the first use case -- which we refer to as ``masking'' -- and the second two use cases -- which we refer to as ``augmentation.'' Masking is in fact a special form of augmentation, but where the purpose is to facilitate the transfer of sensitive information to another party.

The aim of this note is to provide an information-theoretic argument against unconstrained use of synthetic data \textit{augmentation} as a tool for investigators to perform statistical inference or create statistical models. The primary aims of this paper are to (a) propose a definition of synthetic data distributions amenable to mathematical analysis, and (b) explore the relationship between synthetic data augmentation and Fisher information. 

Our definition of synthetic data includes the nonparametric bootstrap as a key example. The reason that synthetic data augmentation is strongly limited in its ability to improve maximum likelihood estimation is the same reason that one cannot improve standard errors by sampling an arbitrary number of observations from a bootstrap distribution: all of the information is contained within the empirically drawn sample.

\section{Abstractly Defining Synthetic Data}

In this section we give a simple, flexible definition of synthetic data. The upshot is that we \textit{already} use and accept synthetic data in many contexts. Their use is typically grounded in the context of specific tasks where their usage can be justified on \textit{a priori} grounds.

\begin{definition}
Let $X = (X_1,\dots, X_n)$ be a sequence of observations from sample space $\Omega$. Let $\Xi$ be another sample space. A $\Xi$-valued synthetic distribution $\mathcal{S}$ is a function $\mathcal{S}: \bigcup_{n\in\mathbb{N}} \Omega^n \to \mathcal{P}(\Xi)$, where $\mathcal{P}(\Xi)$ is the set of probability distributions over $\Xi$.\footnote{More precisely, we assume that $\Xi$ is a \textit{measure space}, not just a set.}
\end{definition}

In other words, a synthetic distribution $\mathcal{S}$ is a function that assigns to any empirical sample drawn from $\Omega$ a probability distribution over the sample space $\Xi$.

Synthetic distributions are ubiquituous in both statistical inference and in machine learning. Perhaps the most famous synthetic distribution is the nonparametric bootstrap.

\begin{example}
Let $X = (X_1,\dots, X_n)$ be a sample from $\Omega$. The bootstrap distribution $\mathcal{B}: \bigcup_{n\in\mathbb{N}} \Omega^n \to \mathcal{P}(\Omega)$ associates to each $X$ the probability mass function
\begin{equation}
\begin{array}{ccc}
    p_{\mathcal{B}(X)}(x) & = & \frac{1}{n}\sum\limits_{i=1}^{n} \mathbb{1}_{X_i}(x)\\ 
    & = & \frac{\#\{x=X_i\,|\,X_i\in \mathbb{X}\}}{n}.
\end{array}
\end{equation}
\end{example}

Similarly, the parametric bootstrap is a synthetic distribution. 

There is no requirement that synthetic distributions be ``realistic'' in the sense that they consistently estimate the original joint distribution of the sample $X$ like the nonparametric bootstrap. In fact, many uses of synthetic data specifically avoid this requirement. For example, reweighting techniques are \textit{also} synthetic distributions.
\begin{example}
Let $X_i = (C_i,Y_i)$ where $C_i$ is a categorical variable. Let $X$ be a dataset. Then define the synthetic distribution $\mathcal{S}$ by first sampling a class $C_i$ uniformly in the sample's observed values of $C$, and then by sampling from the bootstrap on only those with $C_i = C$.

To be more precise, let $\mathbb{X}_C = \{X_i\,|\, C_i = C\}$ be the subset of the sample where $C_i = C$ and $n_C = \#\{C\,|\,\exists i \, C_i = C \}$ be the number of distinct values of $C$ observed in the sample. Then the reweighted probability mass function is given by:
\begin{equation}
    p(C,Y) = \frac{1}{n_C} p_{\mathcal{B}(\mathbb{X}_C)}(C,Y).
\end{equation}
\end{example}

Synthetic distributions can also be used to ensure certain geometric symmetries in the probability distribution that may not be present in the original data, for example by applying iterates of a periodic operator.

\begin{example}
Let $\mathbb{X}$ be a sample, and let $r:\Omega \to \Omega$ be a periodic function with period $n$. Then let 
\begin{equation}
    p(r^k(x)) = \frac{1}{n} p_{\mathcal{B}(X)}(x).
\end{equation}
In other words, first select the iterate of $r$ to apply uniformly in $\{1,\dots, n\}$ and then select $x$ from the bootstrap distribution of the original sample. 

A simple example of this kind of distribution is to let $\Omega$ be a space of square images $x$ with label $y$, and $r$ be rotation by $90^{\circ}$, which has period $4$.

\end{example}

This kind of synthetic distribution may be helpful in improving the generalizability of image classifying models, since it can ensure that concepts like ``sky'' and ``ocean'' aren't tied to their vertical height or orientation in an image.

In each of the above cases, the use of synthetic data might be suitable for specific tasks:
\begin{enumerate}
    \item The theory of the bootstrap for nonparametric estimation and testing is well-developed,
    \item The theory of reweighting to give unbiased estimates of estimands commonly of interest in causal inference is well attested to, and
    \item The use of synthetic data to reduce the risk of overfitting due to geometric artifacts in the sample is common in object detection.
\end{enumerate}

A common thread in these uses is the presence of \textit{a priori} justification for the use of these techniques specific to the task at hand. As we will discuss in the next section, the No Free Lunch theorems and the bias-variance decomposition in machine learning partially explain \textit{why} the use of synthetic data should be justified at the task level. 

\section{Synthetic Data and Fisher Information: A Caution}

In general, synthetic data do not add information to the system to bolster maximum likelihood based inference in recovering the original parameter $\theta$. Intuitively, this is due to how prior knowledge is handled in maximum likelihood techniques. The knowledge of the underlying data generating process is encoded via the choice of parametric family of distributions used to pose the optimization problem in maximum likelihood estimation; namely, $\hat{\theta} = \text{argmax}_{\theta \in \Theta} \mathcal{L}(X)$ for empirical observations $X$. By contrast, synthetic data augmentation leverages prior knowledge in an essentially orthogonal way, using data to construct a distribution on which instances are drawn via simulation to add sample to a database. 

We demonstrate below how synthetic data augmentation (and its mode of incorporating prior knowledge) interplays with maximum likelihood techniques. A synthetic  distribution adds \textit{no} marginal Fisher information to the sequence of empirical observations, while the Fisher information a synthetic distribution can contribute to a disjoint sample is bounded by the Fisher information of the original sample.  


We consider a sequential use of synthetic data: first, $n$ units $X_1,\dots, X_n$ are drawn iid from a distribution $\theta$ for some $\theta \in \Theta$. Then, a synthetic distribution $\mathcal{S}(X)$ is constructed and $S_i\sim \mathcal{S}(X)$ are drawn iid. We wish to evaluate the relationship between the Fisher information contained in the sample $X$ and the information $(X,S)$ under the above sampling procedure.

\begin{theorem} \label{thm:unconditional}
Let $\Theta$ be a space of distributions and $X_i \sim \theta$ be an iid sample from some $\theta\in \Theta$. Let $\mathcal{S}(X)$ be a distribution whose parameters depend only on the sample $X$, and let $S_i\sim \mathcal{S}(X)$ be iid. Then with respect to $\Theta$, we have 
\begin{equation}
    I_{X,S}(\theta) = I_{X}(\theta).
\end{equation}
Equivalently, $I_{S|X}(\theta) = 0$.
\end{theorem}

This theorem demonstrates that the use of synthetic data generates zero \textit{marginal} Fisher information---measured by $I_{S|X}(\theta)$---about the original distribution $\theta$ beyond that contained in the sample $X$ used to generate the synthetic distribution. As we will see, this also entails a bound on the unconditional Fisher information $I_{S}(\theta)$.

\begin{proof}
The joint distribution $q_{\theta}(X,S)$ under this sampling procedure is given by 
\begin{equation}
    q_{\theta}(X,S) = \prod p_{\theta}(X_i) \times \prod q_{\theta|X}(S_i)
\end{equation}
where $p_{\theta}$ is the probability density function of $X$ in $\theta$ and $\prod q_{\theta|X}$ is the density conditional on the sample $X$. Furthermore, by construction of $\mathcal{S}(X)$ the \textit{conditional} density $q_{\theta|X}(S_i) = p_{\mathcal{S}(X)}(S_i)$ does not depend on $\theta$. Hence

\begin{equation}
    q_{\theta}(X,S) = \prod p_{\theta}(X_i) \times \prod p_{\mathcal{S}(X)}(S_i)
\end{equation}

To calculate the Fisher information we decompose the likelihood function as follows:

\begin{equation}
\begin{array}{ccc}
     \mathcal{L}_{XS}(\theta) & = & \mathcal{L}_{X}(\theta)\mathcal{L}_{X|S}(\theta) \\
     & = &\mathcal{L}_{X}(\theta) \times \underbrace{\prod p_{\mathcal{S}(X)}(S_i)}_{c_{X}(S)} \\
     & = & c_{X}(S) \mathcal{L}_{X}(\theta),
\end{array}
\end{equation}

Observe that the term $c_{X}(S)$ is constant in $\theta$, as conditional in $X$ there is no dependence of the probability of $S$ on $\theta$.

Taking logarithms we have 
\begin{equation}
\ell_{X,S}(\theta) = \log(c_{X}(S)) + \ell_{X}(\theta).
\end{equation}
The term $\log(c_{X}(S))$ vanishes when evaluating its partial derivative in $\theta$, so that
\begin{equation}
    \frac{\partial^2}{\partial \theta^2} \ell_{X,S}(\theta) = \frac{\partial^2}{\partial \theta^2} \ell_{X}(\theta)
\end{equation}
and hence by the equality $I_{\widetilde{D}}(\theta) = -\widetilde{E}\left[\frac{\partial^2}{\partial \theta^2} \ell_{\widetilde{D}}(\theta)\right]$ (see \cite{keener_theoretical_2010} equation 4.16)

\begin{equation}
\begin{array}{ccc}
I_{X,S}(\theta) &=& -\widetilde{E}\left[\frac{\partial^2}{\partial \theta^2} \ell_{X,S}(\theta)\right] \\ 
& = & -\widetilde{E}\left[\frac{\partial^2}{\partial \theta^2} \ell_{X}(\theta)\right]\\
& = & I_{X}(\theta).
\end{array}
\end{equation} 

Consequently
\begin{equation} I_{X,S}(\theta)  = I_{X}(\theta)
\end{equation}

By the chain rule of Fisher information (see, for example, \cite{zamir_proof_1998}),  

\begin{equation}
\begin{array}{ccc}
I_{X,S}(\theta) & = &  I_{X}(\theta) + I_{S|X}(\theta)\\
& = & I_{X,S}(\theta) + I_{S|X}(\theta)
\end{array}
\end{equation}
Hence
\begin{equation}
I_{S|X}(\theta)  =  0 
\end{equation}
as desired.
\end{proof}

Thus, drawing extra sample $S$ from the synthetic distribution $\mathcal{S}(X)$ adds no \textit{marginal} Fisher information about the underlying parameter $\theta$. In other words, $I_{S|X}(\theta) = 0$. 


\begin{example} Let $X_i \sim \theta$ be iid.
Observe that for the nonparametric bootstrap $\mathcal{S}(X) = \mathcal{B}(X)$, for every $\theta \in \Theta$ we have that $q_{\theta|X}(S_i) = \frac{1}{n}$. 

Thus, 
\begin{equation}
    \mathcal{L}_{X,S}(\theta) =\frac{1}{n^k} \prod\limits_{i=1}^n p_{\theta}(X_i).
\end{equation}

So in log-likelihood we have 
\begin{equation}
    \ell_{X,S}(\theta) = -k\log(n) + \sum \log(p_{\theta}(X_i)).
\end{equation}
Taking second derivatives of $\ell$ results in the vanishing of the constant $-k\log(n)$ term, so that 
\begin{equation}
I_{X,S}(\theta)  = I_{X}(\theta) = nI_{X_1}(\theta).
\end{equation}
\end{example}

This does \textit{not} mean that, in isolation, $S_i$ sampled from a synthetic dataset $\mathcal{S}(X)$ have zero Fisher information:

\begin{example}
    Let $S_i \sim \mathcal{B}(X)$ with $X_i\sim \theta$ be an iid sample of size $n$. Then 
\begin{equation}
        I_{S}(\theta) \leq nI_{X_1}(\theta)
\end{equation}

That is, \textit{if} all you know about the $S_i$ is that they were generated from a bootstrap resample $\mathcal{B}(X)$ from $\mathcal{D}$ and you know the size of the sample $X$ used to generate the bootstrap, you get an upper bound on $X$. A lower bound on Fisher information is the number of \textit{distinct} units $S_i$ observed.
\end{example}

But what if you are merely the \textit{recipient} of synthetic data $S$ and are not given the original sample $X$? How much information can be gleaned from $S$ about $\theta$? The preceding result and the chain rule for Fisher information can be used to bound the unconditional Fisher information:

\begin{theorem} \label{thm:conditional} Suppose that $X$ were drawn from $\theta \in \Theta$, $\mathcal{S}(X)$ is a synthetic distribution, and let $S_i\sim \mathcal{S}(X)$ be iid. Then the Fisher information in $\theta$ of $S$ satisfies
\begin{equation}
    I_{S}(\theta) \leq I_{X}(\theta)
\end{equation}
\end{theorem}
\begin{proof}
By the chain rule of Fisher information, we have that 
\begin{equation}
     I_{X,S}(\theta) = I_{S}(\theta) + I_{X|S}(\theta).
\end{equation}

By our previous result, we know that $I_{X,S}(\theta) = I_{X}(\theta)$. Since Fisher information is non-negative we can conclude
\begin{equation}
\begin{array}{ccc}
    I_{S}(\theta) & \leq &  I_{S}(\theta) + I_{X|S}(\theta) \\
     & = & I_{X,S}(\theta) \\
     & = & I_{X}(\theta),
\end{array}
\end{equation}
establishing the inequality.
\end{proof}

Thus, the Fisher information contributed by a synthetic sample is bounded above by that contributed by the original sample $X$. There is no free lunch in inference from the use of synthetic data.

Do note, however, that to use the sample from the synthetic dataset to recover information about $\theta$ generally requires knowledge of how $\mathcal{S}$ depends on the sample $X$. For example, if $\mathcal{S}(X)$ does not depend on $X$ (i.e. $\mathcal{S}(X) = \mathcal{N}(0,1)$ for all $X$) then no information about $\theta$ can be recovered. In other words, the lower bound $I_{S}(\theta) \geq 0$ can be sharp when $S_i\sim \mathcal{S}(X)$. 

One practical application of this result is that unless you know the \textit{original} sample size of subjects meeting your inclusion criteria used in the construction of the synthetic data set $\mathcal{S}(X)$ you are not guaranteed to have much Fisher information about your target distribution. In ``masking'' applications of synthetic data generation, this information \textit{is} readily available and helps ground the validity of its use.

This issue is compounded if the data used to train the synthetic data $\mathcal{S}$ is not disjoint from the genuine dataset $X$ that you supplement it with. 
So, not only do you need to know \textit{how much} data from $\theta$ is inside the training set for $\mathcal{S}$, but also its relation to the genuine data $X$ that you already have available to you.

\section{An Analogue in Bayesian Estimation} \label{sec:bayes_vanilla}

An important class of synthetic distributions are distributions arising as the posteriors of some distributions given a (possibly small) amount of data. Suppose that we are given a prior $\mathbb{P}$ and condition it on the genuine data set $X$ to yield the posterior $\mathbb{P}(X)$. This gives a well defined synthetic distribution $\mathcal{S}(X) = \mathbb{P}_{|X}$.

On the Bayesian account, there is no information gain in sampling from the posterior $\mathbb{P'} = \mathbb{P}(\overline{X})$, as described by the Bayesian reflection principle as stated in \cite{goldstein_prevision_1983}. Formally, for every event $E$ and possible observation $Y$
\[\mathbb{E}_{\mathbb{P'}}[\mathbb{P}'(E|Y)] = \mathbb{P}'(E).\]
In other words, the posterior contains \textit{all} possible information about the distribution given the observations $X$, so that conditionally chaining on synthetic data drawn from the same dataset do not change the posterior 
In particular, synthetic observations drawn from $\mathbb{P}'$ yield no information beyond that contained in $\mathbb{P}'$.

\section{Synthetic Data Augmentation Isn't Categorically Bad; It's Just Different}

The above analysis is intended to clarify the relationship between synthetic data \textit{augmentation}---namely, the use of synthetically generated data from some source to boost sample size---in inferential tasks common to frequentist and Bayesian analyses. The analysis indicates that there is not a neat way to integrate the prior knowledge encoded by synthetic distributions into maximum likelihood or Bayesian approaches. This is due to the fact that prior knowledge is already strongly encoded into the techniques of maximum likelihood estimation and the Bayesian prior, respectively. Borrowing terminology from \cite{rodu2023black}, synthetic data augmentation is an \textit{outcome reasoning} approach to encoding \textit{a priori} information in the fitting of models; its use is obviated in traditional statistical modeling.

\begin{example}
Let $X$ be a sample of iid coin flips, Bernoulli(p) with outcomes {H,T}. The prior knowledge the analyst wishes to encode is that there is a censoring mechanism that drops every second H. In the maximum likelihood approach, this information would be used to modify the likelihood function. Similarly, for the Bayesian modeling approach the prior would be modified to encode censoring's impact on the observed data. 

Synthetic data augmentation permits manipulating the analytic data set in order to encode the information, perhaps by augmenting the observed $X$ with additional H.
\end{example}

This is not to say that synthetic data cannot play a role in data analysis; rather, its use must be carefully considered. The use of synthetic data to \textit{mask} sensitive data for research use is a well-understood process that by its very construction warrants the use of standard inferential techniques. 

For synthetic data \textit{augmentation}, better use cases appear to be when the analyst does not have analytic descriptions of how the data generating process maps to the estimation procedure. Two examples of this: 
\begin{enumerate}
    \item The use of synthetic data in human-in-the-loop workflows which are designed to extract information from a human labeler who audits the outputs of the synthetic distribution. In this case, modeling the complex human decision process is bypassed in the modeling step and the information in instead encoded in the augmented data set.
    \item The use of theory-informed synthetic data augmentation to corral a model to be closer to some theoretically justified functional form is a way to reduce variance in the model and flexibly account for prior knowledge in an approximate way. An example of this is rotating images to encode rotation-invariance in an object detection algorithm, where the object detection algorithm is treated as a black box.
\end{enumerate}

In each of these cases, the mode of information extraction using the synthetic distribution is tailored to a specific task. There is no ``off the shelf'' synthetic distribution process $\mathcal{S}$ that warrants its use in all tasks, especially in the context of inference. 

When an analyst is using \textit{outcome reasoning} (roughly speaking: using the common task framework\cite{donoho201750}, leaderboards, or more informally train/test splits to reason about the performance of their algorithms) then synthetic data augmentation may be the only means available for encoding the analyst's beliefs into the analysis. As described in \cite{rodu2023black}, the more traditional statistical forms of reasoning -- \textit{model reasoning}, \textit{warranted reasoning} -- require mathematical statements to encode beliefs. Synthetic data augmentation is not bad; it is just different from the traditional ways of encoding beliefs. Being able to distinguish between forms of reasoning aids the analyst, and critics, in locating where beliefs are being encoded in the analysis.

\subsection{A Bayesian Reframing of Synthetic Data Augmentation}

While \ref{sec:bayes_vanilla} argues that synthetic data augmentation alone cannot coherently shift the posterior distribution when the synthetic data is generated from data already known to the investigator, the most common use case for synthetic data augmentation would be to source synthetic data from a third party to include in the Bayesian posterior update. The specification of the prior does not have to treat the \textit{de novo} data $X$ and synthetically generated data $S$ on the same footing. The prior can incorporate the fact that $X$ and $S$ were generated via different processes and can reflect whatever knowledge or ignorance the investigator has regarding the construction and training of the synthetic distribution from which $S$ was drawn. In notation:
\[ \mathbb{P}(H| (X \text{ is observed real data }) \& \, (S \text{ is synthetically generated data from } \mathcal{S}.))\]
The condition ``$(X \text{ is observed real data }) \& \, (S \text{ is synthetically generated data from } \mathcal{S}.)$'' is highly complex, requiring the prior $\mathbb{P}$ to decide on how to weight the real data $X$ against (a) the degree of belief in the synthetic distribution generating process $\mathcal{S}$ (which could be a black box!) (b) How to weight the limited synthetic sample drawn from $\mathcal{S}$, knowing that that sample may be a bad approximation of $\mathcal{S}$, and (c) how to weight the real data $X$ against the coarse approximation of $\mathcal{S}$ afforded by $S$.

One could take a simple approach and stipulate a prior $\mathbb{P}_{naive}$ that treats the synthetic data \textit{as though} it were on equal footing with real data and set
\[ \mathbb{P}(H| (X \text{ is observed real data } \& \, S \text{ is synthetically generated data from } \mathcal{S}.) := \mathbb{P}_{naive}(H|X \cup S).\]
 In some cases this may be justified on theoretical grounds, as in the case of data masking. But in many possible cases this is plainly bad modeling: if the cost of collecting data is expensive relative to the generation of synthetic data, it becomes all too easy to overwhelm the dataset with increasing number of draws from $\mathcal{S}$. Under $\mathbb{P}_{naive}$, that is \textit{precisely} the cost-minimizing way to update. Moreover, such an endeavor would generally be useless: if $\mathcal{S}$ is in the conjugate family of $\mathbb{P}_{naive}$, such posterior updates will just converge back to $\mathcal{S}$ anyways.
 
 Thus, to avoid falling into this degeneracy, $\mathbb{P}_{naive}$ \textit{cannot} be used. 
 
\subsection{Considerations When Using Synthetic Data From a Third Party}

Distilling insights from previous sections, we highlight two  considerations for an analyst thinking of using synthetic data generated by a third party. Here, we are drawing a distinction between $S$ generated from from a genuine data set $X$ that the analyst has full access to versus $S'$, generated from genuine data $X'$, where the analyst does not have direct access to $X'$. The analyst's intent is to use an analytic data set, $\widetilde{X}$, that is a merging of $\{X,S'\}$. 

The first considerations is that $X$ and $X'$ may not be disjoint because they share observational units. This can arise in situations like when an analyst at a hospital has access to local data, $X$, and is requesting synthetic data, $S'$, from a multi-hospital repository. This centralized repository builds $S'$ using $X'$ and $X'$ contains, possibly all, observational units in $X$. This complicates an accounting of the amount of information obtained by incorporating $S'$ in the analysis.

Second, in most cases, the synthetic data provided by the third party will be a combination of masking and augmentation. One value proposition for getting $S'$ is that $S'$ was constructed using more observational units than the data available to the analyst. Though this use case is not common, an analyst could merge their $X$ with an $S'$ that is merely a masking of the genuine $X'$. At one extreme, when an analysts has no data and receives a masked data set from a third party then the analyst can be said to have augmented the null set with the masked data. 

Before completing the discussion of the second consideration, we discuss of the term \textit{augmentation}. The term augmentation can be thought of as pointing to the additional rows obtained through the analyst's merging of $X$ and $S'$. But there is another meaning for \textit{augmentation} that has been developed in this manuscript's previous sections: synthetic data can be used to encode information about the analyst's beliefs about the data, or of the appropriate features of the model(s) being fit. This kind of information augmentation is common in the object detection literature, where synthetic data are used to incorporate the domain-information that there is often rotation invariance for an image. At one (implausible) extreme, is the possibility that the $S'$ from the third party is generated in a way that perfectly captures the generation process of $X$. Should that situation ever arise then the analyst will be left with only aleatoric uncertainty and analyses will largely reduce to basic data summaries. More realistically, it is incumbent upon the analyst to identify and articulate the domain-information that is usefully added through $S'$.

If synthetic data received from a third party are able to improve an analysis then it is likely through a combination of masking and augmentation. An analyst should be able to articulate what forms of information are entering their analysis from the third party's data set.

\section{Conclusion}

The above argument places strong limits on the use of synthetic data to augment a dataset for the purposes of standard maximum likelihood estimation. Theorem \ref{thm:unconditional} shows that one cannot augment ones \textit{own} data using any synthetic data augmentation technique, while Theorem \ref{thm:conditional} shows that the marginal information obtained through augmentation from a a third party's synthetic data requires knowledge of the underlying Fisher information of the data used to generate \textit{their} synthetic distribution to properly calibrate the total information. Theorem \ref{thm:conditional} does not prohibit using synthetic data techniques to \textit{mask} data for privacy considerations, because the required level of transparency can be achieved through disclosures from the parties \textit{masking} the data. Failing that knowledge, an investigator is liable to get incorrectly reduced standard errors.

We want to be clear that the use of synthetic data to augment predictive models is beyond the scope of this discussion and requires a different set of considerations.

\bibliography{refs}

\end{document}